\documentclass[12pt]{article}
\usepackage{amsfonts}
\usepackage{amssymb}
\usepackage{amsmath}

\newtheorem{theorem}{Theorem}[section]
\newtheorem{proposition}[theorem]{Proposition}
\newtheorem{lemma}[theorem]{Lemma}
\newtheorem{deff}[theorem]{Definition}
\newenvironment{definition}{\begin{deff}\em}{\end{deff}}
\newtheorem{defs}[theorem]{Definitions}
\newenvironment{definitions}{\begin{defs}\em}{\end{defs}}
\newenvironment{proof}{\noindent \textbf{Proof:}}{\hfill 
$\Box$  \vspace{1ex}} 
\newtheorem{ex}[theorem]{Example}

\newtheorem{exs}[theorem]{Examples}
\newenvironment{examples}{\begin{exs}\em}{\end{exs}}
\newtheorem{corollary}[theorem]{Corollary}

\def\Z{\mathbb{Z}}
\def\F{\mathbb{F}}
\def\FF{\mathbf{K}}
\def\no{\mathbb{N}_0}
\def\nor{n-order}
\def\nw{n-weight}
\def\rb{{\mbox{\boldmath$\rho$}}}
\def\br{{\mbox{\boldmath$\rho$}}}

\def\ba{{\boldsymbol{a}}}
\def\bb{{\boldsymbol{b}}}
\def\bc{{\boldsymbol{c}}}
\def\be{{\boldsymbol{e}}}
\def\div{\textrm{div}}

\def\h{\mathcal{S}} 
\def\S{\mathcal{S}}

\def\r{\rho}
\def\setm{\{1, \ldots, m\}}
\def\l{\lambda}
\def\lub{\mathrm{lub}}

\def\ll{\preccurlyeq}

\def\ra{\rightarrow}

\def\M{{\cal M}}
\def\N{{\mathbb N}}
\def\R{\mathbf{R}}

\def\O{\mathcal O}
\def\S{{\cal S}}
\def\P{{\mathbb P}}

\def\B{\mathcal B}
\def\G{\Gamma}
\def\cL{\mathcal L}
\def\U{\mathcal U}
\def\X{\mathcal X}
\def\Y{\mathcal Y}

\def\cL{\mathcal L}
\def\cU{\mathcal U}

\begin{document}

\begin{center}
{\large\textsc{On algebras admitting a complete set of near weights,  evaluation codes and Goppa codes} }\\ \ \\
C\'{\i}cero Carvalho\footnote{Universidade Federal de Uberl\^andia, Faculdade de Matem\'atica, Av. J.N. de \'Avila 2160,  38408-100 Uberl\^andia -- MG, Brazil. email: cicero@ufu.br. Research partially supported by FAPEMIG - grant CEX APQ-4716-5.01/07} and Erc\'{\i}lio Silva\footnote{Universidade Federal do ABC, CMCC, 
Rua Santa Ad\'elia 166, 09210-170 Santo Andr\'e -- SP, Brazil. email: ercilio@ufabc.edu.br.} 
\end{center}
\vspace{6ex}

\noindent
\textbf{Abstract. }{\footnotesize In 1998 H\o holdt, van Lint and Pellikaan introduced the concept of a ``weight function'' defined on a $\F_q$-algebra  and used it to construct linear codes, obtaining among them the algebraic-geometric (AG) codes supported on one point. Later, in 1999, it was proved by Matsumoto that all codes produced using a weight function are actually AG codes supported on one point.
Recently, ``near weight functions'' (a generalization of weight functions), also defined on a $\F_q$-algebra,  were introduced to study codes supported on two points. In this paper we show that an algebra admits a set of $m$ near weight functions having a compatibility property, namely, the set is a ``complete set'', if and only if it is the ring of regular functions of an affine geometrically irreducible algebraic curve defined over $\F_q$ whose  points at infinity have a total of  $m$ rational branches. Then the codes produced using the 
near weight functions are exactly the AG codes supported on $m$ points.
A formula for the minimum distance of these codes is presented with examples which show that in some situations it compares better than the usual Goppa bound.} 
\vspace{2ex}

\vspace{6ex}

\section{Introduction}

In 1981 V.D. Goppa showed how to use algebraic curves to produce error correcting codes (v.\ \cite{goppa}), and his construction opened a new area of research in coding theory. After a decade of studies, researchers started to wonder if it was possible to find a simpler way to produce these (so called) algebraic-geometric, or Goppa, codes, one of the earliest attempt being made by Blahut (\cite{blahut}). In 1998 H\o holdt et al.\ (v. \cite{h-vl-p}) presented a simple construction for error correcting codes, using an $\F$-algebra $\R$ and what they called a {\em weight function} on $\R$,  their construction clearly producing algebraic-geometric codes supported on one point. The theory presented in \cite{h-vl-p} was recently generalized (v.\ \cite{silva} and \cite{4a}) by replacing weight functions by other functions on $\R$, called {\em near weights}. In the present work we 
study specially algebras that admit $m$ near weight functions with the property of being ``a complete set'' (see Definition \ref{completeset}). We will characterize them as being the ring of regular functions of an affine geometrically irreducible algebraic curve whose  points at infinity have a total of  $m$ rational branches, from this we conclude that the codes obtained from such algebras using the complete set of near weight functions are exactly the algebraic-geometric codes supported on $m$ points
(thus generalizing results in \cite{matsumoto} and \cite{mun-tor}).

In what follows we will denote by $\no$ the set of nonnegative integers. Let   
$\F$ be a field and  $\R$  be a commutative ring that contains $\F$, i.e.\ an $\F$-algebra.
Given a function $\rho: \R \ra \no \cup \{- \infty\}$ let $\U_\rho := \{ f \in \R \; | \; \rho(f) \leq \rho(1)\}$ and $\M_\rho:= \{f \in \R \; | \; \rho(f)  >  \rho(1)\}$. 

\begin{definitions}
We call $\rho$ a {\em near order} function on $\R$ (or \nor\ for short) if for any $f, g \in \R$ we have: \\
(N0) $\rho(f) = - \infty$ $\Leftrightarrow$ $f = 0$; \\
(N1) $\rho(\lambda f) = \rho(f) $ $\forall \lambda \in \F^*$; \\
(N2) $\rho(f +  g ) \leq \max \{\rho(f), \rho(g)\}$; \\
(N3) if $\rho(f) < \rho(g) $ then $\rho(f h) \leq \rho(g h) $; if moreover $h \in \M_\rho$ then $ \rho(f h) < \rho(g h) $; \\
(N4) if $\rho(f) = \rho(g)$ and $f, g \in \M_\rho$ then there exists $ \lambda \in \F^*$ such that $\rho(f - \lambda g) < \rho(f)$.

An \nor\ function is called a {\em near weight} (or \nw\ for short) if it also satisfies the following condition. 

\noindent
(N5) $\rho(f g ) \leq \rho(f) + \rho(g)$ and equality holds when $f, g \in \M_\rho$.
\end{definitions}

A trivial way to define an \nor\ $\rho$ on $\R$ is to set $\r(0) := - \infty$ and $\r(f) = 1$ for all $f \in \R$, $f \neq 0$, then we have $\M_\r = \emptyset$ and $\U_\r = \R$. We want to avoid such functions, 
so we say that an \nor\ $\rho$ is {\em trivial} if $\M_\r = \emptyset$ and from now on we work only with nontrivial \nor\ functions.

From (N3) it follows that $\M_\r$ does not have zero divisors, 
we also get the following results (cf. \cite[Lemma 4]{4a}).

\begin{lemma} Let $\r$ be an \nor\ on $\R$, then: \\
i) the element $\lambda$ in (N4) is uniquely determined; \\
ii) if $\rho(f) \neq \rho(g)$ then $\r(f + g) = \max\{\r(f), \r(g)\}$.
\end{lemma}

\noindent
\textbf{Notation.} In the next sections we deal with subsets of $\no^m$, and will use the following conventions: we  denote by $\mathbf{0}$ the $m$-tuple having all entries equal to zero; when we write $\ba \in \no^m$ it's to be understood that the entries of the $m$-tuple $\ba$ are $\ba := (a_1, \ldots, a_m)$ (similarly for $\bb, \bc \in \no^m$); we write sometimes $\ba_i \in \no^m$, being then understood that $\ba_i = (a_{i 1}, \ldots,a_{i m})$. Also, for $i \in \{1, \ldots, m\}$  we denote by  $\be_i$ the $m$-tuple that has all entries equal to zero, except the $i$-th entry, which is equal to 1. We add $m$-tuples and multiply them by nonnegative integers in the usual way.

\section{Codes from near weights}

In this section we show how to construct codes from algebras that admit a complete set of \nw s, and give a lower bound for their minimum distance. We begin by introducing the concept of normalized \nor s.

\begin{definition} Let $\r$ be an \nor\ function, we define the {\em normalization} $\rho'$ of $\r$ as being the function $\r' : \R \rightarrow \no \cup \{- \infty\}$ defined by $\r'(0) = - \infty$, $\r'(f) := 0$ if $f \in \U_\r \setminus \{0\}$ and $\r'(f) := \r(f)$ if $f \in \M_\r$. 
\end{definition}

From the proof of \cite[Proposition 1]{4a} we know that $\r'$ is an \nor, $\U_{\r'} =  \U_\r$ and $\M_{\r'} =  \M_\r$. From now on we work only with normalized \nor s. If $\r$ is an \nw\ then from (N5) we see that $\U_\r$ is a subalgebra of $\R$. 

In this section we will  show how to construct linear codes from $\F$-algebras and a set of \nw s which have a compatibility property which we define now. Let $\{\r_1, \ldots, \r_m\}$ be a set of (nontrivial,  normalized) \nw s. 

\begin{definition} \label{completeset} We say that $\{\r_1, \ldots, \r_m\}$ is a {\em complete set of \nw s for $\R$} if $\cap_{i = 1}^m \U_{\r_i} = \F$ and for all $k \in \{1, \ldots, m\}$ we have that  $\no \setminus \rho_k(\cap_{1\leq i \leq m\, ; \,  i \neq k} \cU_{\r_i})$ is a finite set.
\end{definition}

Let $\R$ be an $\F$-algebra that admits $\{\r_1,\ldots, \r_m\}$ as a complete set of \nw s.
Given $\ba = (a_1, \ldots,  a_m)  \in \no^m$  we define 

\[\cL(\ba) :=\{ f \in \R : \rho_i(f) \leq a_i \;  \forall \, i = 1 ,\ldots, m\}.
\] 

From (N0),(N1) and (N2) we get that $ \cL$ is an $\F$-vector subspace of $\R$.

\begin{lemma} \label{lemma_dim} For any  $k \in \{1, \ldots, n\}$ we get $\cL(\ba) \subset \cL(\ba + \be_k)$; moreover $\dim (\cL(\ba + \be_k)/ \cL(\ba)) \leq 1$.
\end{lemma}
\begin{proof} Assume that $f, g  \in  \cL(\ba + \be_k) \setminus \cL(\ba)$, from (N4) we know that there exists $\lambda \in \F^*$ such that $\rho_k(f - \lambda g) \leq a_k$,  and from (N1) and (N2) we get $\rho_i(f - \lambda g) \leq a_i$ for all $i \in \{1, \ldots,  m\}\setminus\{k\}$. Thus $f = \lambda g + h$ with $h \in \cL(\ba)$ hence $ \overline{f} = \lambda \overline{g}$ as elements of $\cL(\ba + \be_k)/ \cL(\ba)$.
\end{proof} 

Since $\cL( \mathbf{0}) = \F$ we get as a corollary of the above lemma that $\cL(\ba)$ is an $\F$-vector space of finite dimension for any $\ba \in \no^m$.

For the remainder of this section, we will assume that $\F$ is a finite field.
Let $\varphi : \R \rightarrow \F^n$ be a surjective morphism of $\F$-algebras and let $\ba \in \no^m$.  We will denote by $C(\ba)$ the code $\varphi(\cL(\ba))$ and we want to determine a lower bound for the minimum distance of $C(\ba)^{\perp}$, in a way similar to that which has been done by H\o holdt et  alli in the case where $m = 1$ (cf. \cite[Section 4]{h-vl-p}). 

\begin{definition} \label{defna}
Let $k \in\{1, \ldots, m\}$, and define $N_k(\ba)$ as a set of pairs of functions $\{ (f_{k, 1}, g_{k, 1}), \ldots, (f_{k, \ell_k}, g_{k, \ell_k}) \}$ such that: \\
a) $f_{k, i}, g_{k, i} \in \cL(\ba + \be_k)$ for all $i = 1, \ldots, \ell_k$; \\
b) $\r_k(f_{k, i}) + \r_k(g_{k, i}) = a_k + 1$; \\
c) $\r_k(f_{k, 1}) < \cdots < \r_k(f_{k, \ell_k})$ (hence $\r_k(g_{k, 1}) > \cdots > \r_k(g_{k, \ell_k}) )$; \\
d) given $s \in \{ 1, \ldots, \ell_{k} - 1\}$ we have $f_{k, s} g_{k, r} \in \cL(\ba)$ for all $r = s + 1, \ldots, \ell_k$. \\
We will write $\nu_k(\ba) := \# N_k(\ba)$.
\end{definition}

Now, consider the matrices $M$ and $N$,  where the first $\ell_k$ rows of $M$ are $\varphi(f_{k, 1}),$ $ \ldots,$ $  \varphi(f_{k, \ell_k})$, the first $\ell_k$ columns of $N$ are  $\varphi(g_{k, 1}), \ldots, \varphi(g_{k, \ell_k})$, and we complete the rows of $M$ and the columns of $N$ in a way such that $rank(M) = rank(N) = n$. Let $\mathbf{y} = (y_1, \ldots, y_n) \in \F^n$ and let $D(\mathbf{y}) := (a_{i\, j})_{n \times n}$ where $a_{i\, j} = 0$ if $i \neq j$ and $a_{i\,  i} = y_i$ for $i = 1, \ldots, n$. Since $rank(M) = rank(N) = n$ we get $rank(M D(\mathbf{y}) N) = wt(\mathbf{y})$; moreover if $r, s \in \{1, \ldots,  \ell_k\}$ then $(M D(\mathbf{y}) N)_{r,  s} = \mathbf{y} \cdot (\varphi(f_{k,  r}) * \varphi(g_{k,  s}))$, where $\cdot$ is the usual inner product in  $\F^n$ and $*$ is the usual componentwise product that makes $\F^n$ an $\F$-algebra.

\begin{proposition} If\/  $\mathbf{y} \in C(\ba)^\perp \setminus C(\ba + \be_k)^\perp $ then $rank(M D(\mathbf{y}) N) \geq \# N_k(\ba)$.
\end{proposition}
\begin{proof} We have already noted that $(M D(\mathbf{y}) N)_{r,  s} = \mathbf{y} \cdot \varphi(f_{k,  r} \, g_{k,  s})$ for all $r,  s \in \{1, \ldots,  \ell_k\}$. From definition \ref{defna} (d) we get that the $\ell_k \times \ell_k$ minor at the upper left corner of $M D(\mathbf{y}) N$ is a lower triangular matrix. Since $f_{k, s}  g_{k, s} \in  \cL(\ba + \be_k) \setminus \cL(\ba)$ from Lemma \ref{lemma_dim} we get $\dim \cL(\ba + \be_k) =  \dim \cL(\ba) + 1$ hence $\mathbf{y} \cdot \varphi(f_{k, s} \, g_{k,  s}) \neq 0$ for all $s = 1, \ldots, \ell_k$.
\end{proof}

\begin{definition}
Let $\ba,\bb \in \no^m$ be such that $a_i \leq b_i$ for all $i = 1, \ldots, m$. 
We call a {\em path from $\ba$ to $\bb$} a finite sequence of $m$-tuples ${\cal P}:= (\ba_0, \ba_1, ..., \ba_r)$,  where $\ba_i \in \no^m$ for all $i \in \{0, \ldots, r\}$,   $\ba_0 = \ba$,  $\ba_r  = \bb$ and for any $i \in \{0, \ldots, r -1\}$ we have  $\ba_{i + 1}  = \ba_i + \be_{p(i)}$  for some $p(i) \in \{1, \ldots, m\}$ which is called the {\em step place} of $\ba_i \in {\cal P}$.
\end{definition}

\begin{lemma} \label{lemma_dim_n} Let $\ba \in \no^m$, then there exists $\bb \in \no^m$ such that $\dim C(\bb) = n$ and $a_i \leq b_i$ for all $i \in \{1, \ldots, m\}$.
\end{lemma}
\begin{proof} Since $\varphi$ is surjective there are $f_1, \ldots , f_n \in \R$ such that $\{\varphi(f_1), \ldots, \varphi(f_n)\}$ is a basis for $\F^n$, so it suffices to take $b_i := a_i + \max \{\r_i(f_1), \ldots, \r_i(f_n)\}$, where $i \in \{1, \ldots, m\}$ and set $\bb := (b_1, \ldots, b_m)$.
\end{proof}

As a consequence of the above results we have the following bound for the minimum distance of $C(\ba)^\perp$.

\begin{corollary} \label{thm_bound} Let $\ba \in \no^m$ and let $\bb \in \no^m$ be such that $a_i \leq b_i$ for all $i \in \{1, \ldots, m\}$ and $\dim C(\bb) = n$. Given a path ${\cal P} = (\ba_0, \ldots, \ba_r)$ from $\ba$ to $\bb$ 
the minimum distance of $C(\ba)^\perp$ is bounded from below by $\min \{ \nu_{p(i)}(\ba_i) \; | \;  i = 0, \ldots, r - 1 \}$. 
\end{corollary}

At first glance a major drawback of the above result is that it depends on finding $\bb \in \no^m$ such that $\dim \varphi(\cL(\bb)) = n$, while we would like a bound that does not depend on the knowledge of $\varphi$. The following considerations show that we do not have to find such $\bb$ in order to calculate a bound. 

Let $k \in \{1, \ldots, m\}$, from (N5) we get that $\S_k := \rho_k(\cap_{1\leq i \leq m; i \neq k} \cU_{\r_i})$ is a subsemigroup of $\N_0$ and since $\{\r_1, \ldots, \r_m\}$ is a complete set for $\R$ we get $\#(\no \setminus \S_k) < \infty$ (i.e.\ $\S_k$ is a numerical semigroup). Observe also that given  $\ba \in \no^m$ and $t_1, t_2 \in \S_k$ such that $t_1 + t_2 = a_k + 1$ then taking $f_1, f_2 \in \cap_{1\leq i \leq m; i \neq k} \cU_{\r_i} $ such that 
$t_i = \r_k(f_i)$, $i = 1, 2$, we get $(f_1, f_2) \in  N_k(\ba)$ (of course also  $(f_2, f_1) \in  N_k(\ba)$, if $t_2 \neq t_1$).

\begin{lemma} Let $S$ be a numerical subsemigroup of $\N_0$ of genus $g$, let $c$ be the conductor of $S$ and let $u \in \N_0$. If  $N := \{(a, b) : a, b \in S \setminus\{0\}; a + b = 2 c + u \}$ then $\# N = 2(c - g) + u - 1$. 
\end{lemma}
\begin{proof} We have $2 (c - 1 - g) $ pairs $(a,b) \in N$ such that either $1 \leq a \leq c - 1$ or $1 \leq b \leq c - 1$. And we have $u + 1$ pairs $(a, b) \in N$ with $c \leq a, b \leq c + u$.
\end{proof}

Let $\ba \in \no^m$ and let $(\ba_i)_{i \in \no} \in \no^m$ be a sequence of $m$-tuples such that $\ba = \ba_0$, $\ba_{i + 1} = \ba_i + \be_{p(i)}$ for some $p(i) \in \{1, \ldots, m\}$ and $\lim_{i \rightarrow \infty} a_{i j} = \infty$,  for all $j \in \{1, \ldots , m\}$ and all $i \in \no$. From (the proof of) Lemma \ref{lemma_dim_n} we know that there exists $r \in \no$ such that $\dim C(\ba_r) = n$.  
For $k \in \{1,\ldots, m\}$,  let $c_k$ be the conductor of $\S_k$,  to calculate the bound indicated in Corollary  \ref{thm_bound}  we should calculate $\nu_{p(i)}(\ba_i)$ for all $i \in \{0, r-1\}$, but we observe that: \\
a) if $\nu_k (\ba) =: h > 2 (c_k - g_k) - 1$ then  set $u := h - 2(c_k - g_k) + 1$;  we shall calculate $\nu_k(\ba_i)$ at most for $m$-tuples $\ba_i$ such that $a_{i k} \leq 2 c_k + u - 1$ (in fact, if $a_{i\, k} + 1 > 2 c_k + u$ we will get $\nu_k (\ba_i) > h$);\\
b) if $\nu_k (\ba) \leq 2 (c_k - g_k) - 1$  then we shall calculate $\nu_k(\ba_i)$ at most for  $m$-tuples $\ba_i$ such that $a_{i\, k} \leq 2 c_k - 1$ (in fact,  if $a_{i\, k} + 1 > 2 c_k$ then $\nu_k (\ba_i) > 2 (c_k - g_k) + 1$).

Thus we do not have to know $r$ to calculate the bound.

The next result shows that geometric Goppa codes supported in $m$ points are instances of the codes we described above.

\begin{theorem} \label{ida} Let $\X$ be a nonsingular, geometrically irreducible, projective algebraic curve defined over $\F$, and let $G := \sum_{i - 1}^m a_i Q_i$ and $D := P_1 + \cdots + P_n$ be divisors on $\X$ such that $\textrm{supp}(G) \cap \textrm{supp}(D) = \emptyset$ and $P_i$ is a rational point, for all $i = 1, \ldots, n$ (hence the Goppa code $C_{\cal{L}}(D,G)$ is the set of $m$-tuples $(h(P_1), \ldots, h(P_m))$, where $h \in L(G)$). Then taking $\R := \cap_{Q \in \X\, ; \, Q \neq Q_1, \ldots, Q_m} \O_Q$, where $\O_Q$ is the local ring at $Q \in \X$, and defining $\varphi(f) := (f(P_1), \ldots, f(P_n))$ there exists a complete set of $m$ near weights on $\R$ such that  $C_{\cal{L}}(D,G) = C(\ba)$, where $\ba := (a_1, \ldots, a_m)$ .
\end{theorem}
\begin{proof} Observe that $\R$ is the $\F$-subalgebra of $\F(\X)$ consisting of the functions regular on $\X' := \X \setminus \{Q_1,  \ldots, Q_m\}$. Denoting by $v_k$ the discrete valuation of $\F(\X)$ associated to $Q_k$ ($k \in \{1, \ldots,m\}$), one easily checks that the function $\r_k : \R \rightarrow \no \cup \{\infty\}$  defined by  $\r_k(0)  = - \infty$, $\r_k(f) = 0$ if $v_k(f) \geq 0$ and $\r_k(f) = - v_k(f)$ if $v_k(f) < 0$, for all $f \in \R \setminus \{0\}$ is an \nw\ for all $k \in \setm$. We have $\U_{\r_k} = \R \cap \O_{Q_k}$ and $\M_k = \R \setminus \O_{Q_k}$ for all $k \in \{1, \ldots,  m\}$. Moreover, since $\cap_{Q  \in \X} \O_Q = \F$ (because $\X$ is geometrically irreducible) and $\S_k := \rho_k(\cap_{1\leq i \leq m\, ; \,  i \neq k} \cU_{\r_i}) = \r_k(\cap_{Q  \in \X, \, Q \neq Q_k} \O_Q) $ 
is the Weierstrass semigroup at $Q_k$ for all $k \in \setm$ (hence it has finite genus) we get that $\{\r_1 , \ldots, \r_m\}$ is a complete set of \nw s for $\R$.

Denoting by $M_{P_i}$ the maximal ideal of $\O_{P_i}$ we get that $\R/(M_{P_i} \cap \R) \cong  \O_{P_i}/M_{P_i}$ for all $i \in \{1,  \ldots, n\}$ (see e.g.\ \cite[Prop.\ III.2.9]{sti}), hence $\F^n \cong \R/(M_{P_1} \cap \R) \times \cdots \times \R/(M_{P_n} \cap \R)$ and from the Chinese Remainder Theorem $\varphi$ is an epimorphism. We also have 
$\cL(\ba) = \{ f \in \R \; | -v_k(f) \leq a_k, k = 1, \ldots, m\} =  \{ f \in \F(\X)^* \; | \; \div(f) + \sum_{i = 1}^m a_i Q_i \geq 0\} \cup \{0\} = L(G)$ hence $C(\ba) = C_{\cal{L}}(D,G)$. 
\end{proof} 

Now we present examples which show that when applied to a geometric Goppa code, the bound for the minimum distance found above may be better than the usual Goppa bound.

\begin{examples}
 Let $\X$ be the hermitian curve given by $Y^3 Z +  Y Z^3 - X^4 = 0$ and defined over the field $\F_{3^2}$. Take $Q_1$, $Q_2$ and $Q_3$ to be three distinct rational points, say the  points in the intersection of $\X$, the open set $Z \neq 0$ and the line $X = 0$, let $\ba := (a_1,a_2, a_3) \in \no^3$ and denote by ${\cal C}_{\cL}(D,G_{\ba})$ the geometric Goppa code associated to the divisors $G_{\ba} := a_1 Q_1 + a_2 Q_2 + a_3 Q_3$ and $D = P_1 + \cdots + P_n$, where $P_1, \ldots, P_n$ are distinct rational points, different from $Q_1$, $Q_2$ and $Q_3$. The genus of $\X$ is $3$ and the so-called Goppa bound for the code ${\cal C}_{\cL}(D,G_{\ba})^{\perp}$ is $d_{\ba}:= \deg G_{\ba} - (2 g - 2) = \sum_{i =1}^3 a_i - 4$. Note that $\S_i$ is the semigroup generated by 3 and 4, so  the conductor is 6, for all $k = 1, 2, 3$. To find a bound as described in Corollary \ref{thm_bound} it is useful to know the set $\{(\r_1(f), \r_2(f),\r_3(f)) \in \no^3\;  | \; f \in \R\}$, which in this case is exactly the Weierstrass semigroup ${\cal W}$ associated to $\{Q_1, Q_2, Q_3 \}$, i.e.\ the set $\S = \{(n_1, n_2 , n_3) \in \no^3 \; | \; \div_{\infty} (f) = n_1 Q_1 + n_2 Q_2 + n_3 Q_3 \}$, where $\div_{\infty}(f)$ denotes the pole divisor of $f$. Such semigroups have been much studied in the last decade (see e.g.\ \cite{kim}, \cite{homma-kim}, \cite{gretchen}, \cite{car-tor}), and in \cite{gretchen} we find an explicit description of a generating set for this semigroup, so that 
we may decide if an element of $\no^3$ is or is not in $\S$. Thus,  given $\ba \in \no^3$ we proceed as follows.  For $k = 1, 2, 3$ we calculate $\nu_k(\ba)$, if $\nu_k(\ba) > 2 (6 - 3) - 1 = 5$ then set $A_k := 2\cdot 6  + (\nu_k(\ba) - 5) - 1$, if $\nu_k(\ba) \leq 5$ the we set $A_k := 2\cdot 6 - 1 = 11$. 
Let $r := \sum_{i = 1}^3 (A_i - a_i)$ and consider the path ${\cal P}$ from $\ba$ to $ (A_1, A_2, A_3)$ given by $(\ba_0, \ldots, \ba_r)$ where $\ba_0 = \ba$, $\ba_r = (A_1,  A_2, A_3)$, $\ba_i = \ba + i \be_1$, for $i \in \{ 1,  \ldots, A_1 - a_1 \}$, $\ba_{A_1 - a_1 + j}  = \ba + (A_1 - a_1) \be_1 + j \be_2$, for $j \in \{ 1, \ldots, A_2 - a_2\}$,  and $\ba_{A_1 - a_1 + A_2 - a_2 + k} = \ba + (A_1 - a_1) \be_1 + (A_2 - a_2) \be_2 + k \be_3$, for $k \in \{1,  \ldots, A_3 - a_3\}$. 
From the considerations that precede these examples we get that 
 $\delta_{\ba} := \min \{ \nu_{p(i)}(\ba_i) \; | \;  i = 0, \ldots, r -1 \}$ is a bound for the minimum distance of ${\cal C}_{\cL}(D,G_{\ba})^{\perp}$.

Let $ \rb(N_k(\ba)) := \{ ( (\r_1(f_{k,i}), \r_2(f_{k,i}), \r_3(f_{k,i})), (\r_1(g_{k,i}), \r_2(g_{k,i}), \r_3(g_{k,i})) ) \; | \; i = 1, \ldots, \ell_k \}$, 
taking $\ba = (2, 1, 1)$ we have $\nu_1(\ba) = 2$ (with  $\rb(N_1(\ba)) = \{ ((0,0,0),$ $ (3,0,0)), ((3,0,0), (0,0,0)) \})$, $\nu_2(\ba) = 2$ (with $\rb(N_2(\ba)) = \{ ((0,0,0), (0,3,0)), $ $ ((0,3,0), (0,0,0)) \}$),  and $\nu_3(\ba) = 3$ (with $\rb(N_3(\ba)) = \{ ( ( 0, 0, 0 ), ( 0, 2, 2 ) ) , $ \linebreak $ ( ( 1, 1, 1 ), ( 1, 1, 1 ) ), ( ( 0, 2, 2 ), ( 0, 0, 0 ) )
   \}$).

Thus  $(A_1, A_2,  A_3) = (11, 11, 11)$ and inspecting  ${\cal W}$ we get  that $\delta_\ba = 2$,  while $d_\ba = 0$. In the table below we present results for this and other values of $\ba$. \\

\begin{table}[!h]
\begin{center}
\begin{tabular}{ccccc}   $\ba$ & $(\nu_1(\ba), \nu_2(\ba), \nu_3(\ba))$ & $(A_1, A_2,  A_3)$ & $\delta_\ba$ & $d_\ba$ \\
$(2,1,1)$ & (2,2,2) & (11,11,11) & 2 & 0 \\
$(1,2,1)$ & (2,2,2) & (11,11,11) & 2 & 0 \\
$(1,1,2)$ & (2,2,2) & (11,11,11) & 2 & 0 \\
$(2,2,1)$ & (2,2,3) & (11,11,11) & 2 & 1 \\
$(2,1,2)$ & (2,3,2) & (11,11,11) & 2 & 1 \\
$(1,2,2)$ & (3,2,2) & (11,11,11) & 2 & 1 \\
$(2,2,2)$ & (3,3,3) & (11,11,11) & 2 & 2 \\
$(3,2,2)$ & (4,4,4) & (11,11,11) & 3 & 3 \\
$(2,3,2)$ & (4,4,4) & (11,11,11) & 3 & 3 \\
$(2,2,3)$ & (4,4,4) & (11,11,11) & 4 & 3 
\end{tabular} 
\end{center}
\caption{Bounds for $\delta_\ba$ and $d_\ba$; code $C(\ba)^\perp$;  curve $Y^3 Z +  Y Z^3 - X^4 = 0$, defined over $\F_{9}$.}
\end{table}
 
We also present a similar table, containing examples of codes from the hermitian curve given by $Y^4 Z +  Y Z^4 - X^5 = 0$ and defined over the field $\F_{16}$.  Again, we take 
$Q_1$, $Q_2$ and $Q_3$ to be three distinct rational points of the curve, now ${\cal S}_i$ is the semigroup generated by 4 and 5, so the conductor is 12, for $i = 1, 2, 3$; the genus of the curve is 6. \\

\begin{table}[!h]
\begin{center}
\begin{tabular}{ccccc}   $\ba$ & $(\nu_1(\ba), \nu_2(\ba), \nu_3(\ba))$ & $(A_1, A_2,  A_2)$ & $\delta_\ba$ & $d_\ba$ \\
$(1,2,3)$ & (2,2,2) & (23,23,23) & 2 & -4 \\
$(3,1,3)$ & (2,2,2) & (23,23,23) & 2 & -3 \\
$(3,2,3)$ & (2,2,2) & (23,23,23) & 2 & -2 \\
$(3,3,3)$ & (2,2,2) & (23,23,23) & 2 & -1 \\
$(4,3,2)$ & (2,2,2) & (23,23,23) & 2 & -1 \\
$(4,3,3)$ & (2,2,2) & (23,23,23) & 2 &  0 \\
$(4,4,3)$ & (2,2,3) & (23,23,23) & 2 &  1 
\end{tabular}
\end{center}
\caption{Bounds for $\delta_\ba$ and $d_\ba$; code $C(\ba)^\perp$; curve $Y^4 Z +  Y Z^4 - X^5 = 0$; defined over $\F_{16}$. }
\end{table}

\end{examples}

\section{Algebras with near weights and algebraic curves}

In this section we present a characterization for algebras which admit a complete set of \nw s.

\begin{lemma} \label{1st_lemma} Let $\R$ be an $\F$-algebra and $\r$ an \nw. Let $f, g \in \R$ be such that $\r(f) > 0$, $\r(g) = 0$,  $g \notin \F$ and $\r(f g) < \r(f)$. Then for any $\lambda \in \F^*$ we have $\r(f (g + \lambda) ) = \r(f)$ and $\r(g + \lambda) = 0$.
\end{lemma}
\begin{proof} Let $\lambda \in \F^*$, then $\r(f (g + \lambda)) = 
\r(f g + \lambda f )) \leq \max \{ \r(f g) ,  \r(f)\}$. Since $\r(f g) < \r(f)$ we get $\r(f (g + \lambda)) = \r(f)$ . We also have $g + \lambda \in \U_\r$ since $\U_\r$ is an $\F$-subalgebra of $\R$.
\end{proof}

Let $\R$ be an $\F$-algebra which admits a (not necessarily complete) set of \nw s $\{\r_1, \ldots, \r_m\}$. Let $\rb: \R \setminus \{0\} \ra \no^m$ be the map defined by $\rb(f) :=  (\r_1(f), \ldots, \r_m(f))$ and let $\h_{\r_1,\ldots, \r_m} =\h := \rb(\R \setminus \{0\})$.

We will always assume that if the field $\F$ is finite then $\#(\F) \geq m$.

\begin{definition} Let $\ba_i \in \no^m$,  with $i = 1, \ldots, r$. We define the {\em least upper bound of} $\ba_1, \ldots, \ba_r$ as being the $m$-tuple $\lub(\ba_1, \ldots, \ba_r) := (b_1, \ldots, b_m)$ where $b_j = \max \{a_{j 1}, \ldots, a_{j r}\}$ for all $j = 1, \ldots, m$. 
\end{definition}

\begin{proposition} \label{prop_lub} Let $\ba_1, \ldots, \ba_r \in \h$,  then $\lub(\ba_1, \ldots ,\ba_r) \in \h$. Furthermore, if $f_1, \ldots, f_r \in \R$ are such that $\br(f_i) = \ba_i$ for all $i \in \{1, \ldots, r\}$ then there exists $f \in \R$, $f = \sum_{i = 1}^r \l_i f_i$, where $\l_1, \dots, \l_r \in \F$ such that $\br(f) = \lub(\ba_1, \ldots, \ba_r)$.
\end{proposition}
\begin{proof} Since $\lub(\ba_1,  \ldots, \ba_j) = \lub(\lub(\ba_1, \ldots, \ba_{j - 1}),  \ba_j)$ for all $j = 2, \ldots, r$ it suffices to prove the case where $r = 2$. Let $f ,g \in \R$ be such that $\br(f) = \ba_1$ and $\br(g) = \ba_2$. If $\ba_1 = \ba_2$ then the result is trivial, so we will assume that $\ba_1 \neq \ba_2$, a fortiori $f \neq \lambda g$ for all $\l \in \F^*$. If $\#(\{j \; | \; a_{1 j} = a_{2 j}\}) = m - 1$ then $\lub (\ba_1, \ba_2) \in \{ \ba_1, \ba_2\}$, so we assume that $\#(\{j \; | \; a_{1 j} = a_{2 j}\}) \leq m - 2$.
Let $i \in \setm$, if $\r_i (f) \neq \r_i(g)$ then $\r_i(f + \l g) = \max \{ \r_i (f), \r_i(g)\}$ for all $\l \in \F^*$;  if $\r_i(f) = \r_i(g) = 0$ then for all $\l \in \F^*$ we get $\r_i(f + \l g) = 0$; if $\r_i(f) = \r_i(g) \neq 0$ then there exists a unique $\l_i \in \F^*$ such that $\r_i (f - \l_i g) < \r_i(f)$, hence for all $\l \in \F^*$, $\l \neq - \l_i$ we get $\r_i(f + \l g) = \r_i(f)$. Since $\#(\F) - 1 >  m - 2$ there exists $\l \in \F^*$ such that $\r_i(f + \l g) = \max \{\r_i(f) , \r_i(g)\}$ for all $i \in \setm$.
\end{proof}

\begin{lemma} \label{lemma_dif} Let $\ba$ and $\bb$ be distinct elements of $\h$ and suppose that $a_j = b_j $ for some $j \in \{1, \ldots, m\}$. Then there exists $\bc \in \h$ such that: \\
i) $c_i = \max \{a_i , b_i\}$ for $i \neq j$ and $a_i \neq b_i$; \\
ii) $c_i \leq a_i$ for all $i \neq j$ and $a_i = b_i$; \\
iii) $c_j = a_j = 0$ or $c_j < a_j$.
\end{lemma}
\begin{proof} Let $f, g \in \R$ such that $\br(f) = \ba$ and $\br(g) = \bb$. If $a_j = b_j = 0$ then it suffices to take $\bc = \br(f + g)$. If $a_j = b_j > 0$ then $f, g \in \M_{\r_j}$ and there exists $\l \in \F^*$ such that $\r_j (f - \l g) < a_j$,  so we take $\bc = \br(f - \l g)$.
\end{proof}

Let $\ll$ be the (partial) ordering in $\no^m$ given by the relation $\ba \ll \bb$ if $a_i \leq b_i$ for all $i \in \setm$. 

\begin{proposition} \label{minimal} Let $\ba \in \h$, then the following assertions are equivalent: \\
i)  $\ba$ is a minimal element of the set $\{ \bc \in \h \; | \; c_k = a_k\}$ for some $k \in \setm$ such that $a_k > 0$; \\
ii) $\ba$ is a minimal element of the set $\{ \bc \in \h \; | \; c_i = a_i\}$ for all $i \in \setm $ such that $a_i > 0$.
\end{proposition}
\begin{proof} Assume that $\ba$ is a minimal of the set $\{ \bc \in \h \; | \; c_k = a_k\}$ for some $k \in \setm$ and suppose that $\ba$ is not a minimal of the set $\{ \bc \in \h \; | \; c_j = a_j\}$ for some $j \in \setm$. Then there exists $\bb \in \h$ such that $\bb \ll \ba$, $\bb \neq \ba$ and $b_j = a_j$, furthermore, from the hypothesis we must have $b_k < a_k$. From Lemma \ref{lemma_dif} there exists $\bc \in \h$ such that $c_i \leq \max \{a_i, b_i\}$ for all $i \in \setm$, $c_k = a_k $ and $c_j < a_j$, so $\ba $ is not a minimal of the set $\{ \bc \in \h \; | \; c_k = a_k\}$, a contradiction.
\end{proof}

\begin{definition} If $\ba \in \h$ is a minimal element of the set $\{ \bc \in \h \; | \; c_k = a_k\}$ for some $k \in \setm$ we say that $\ba$ is {\em a minimal} of $\h$ (cf.\ \cite[Section 2]{gretchen}). We will  denote by $\Gamma$ the set of all minimals.
\end{definition}

Observe that $\mathbf{0}$ and the points of $\h$ which have all entries but one equal to zero are minimals.

\begin{theorem} The set $\h$ is a subsemigroup of $\no^m$.
\end{theorem}
\begin{proof} Let $\ba, \bb \in \h$ and let $f, g \in \R$ be such that $\br(f) = \ba$ and $\br(g) = \bb$. Set $\bc := \br(f g)$, for $i \in \setm$ we have $c_i \leq a_i + b_i$ and equality holds whenever $a_i > 0$ and $b_i > 0$, hence $\ba + \bb = \lub(\ba, \bb, \bc)$.
\end{proof}

We assume from now on that $\{\r_1, \ldots, \r_m\}$ is a complete set of \nw s for $\R$; the next result shows that  
 $\Gamma$ generates the semigroup $\h$ under the operation $\lub$.

\begin{lemma} \label{lemma_lub} Let $\ba \in \h$ and let $r$ be the number of nonzero entries of $\ba$, then there exist $\ba_1, \ldots , \ba_r \in \Gamma$ such that $\ba = \lub(\ba_1, \ldots, \ba_r)$. 
\end{lemma}
\begin{proof} Let $\ba \in \h \setminus \Gamma$ and let $\Lambda \subset \setm$ be the set of indexes $i$ for which $\ba_i > 0$; from Proposition \ref{minimal}  $\ba$ is not a minimal in any set $\{ \bb \in \h \; | \; b_i = a_i \}$ with $i \in \Lambda$, then for all $i \in \Lambda$ there exists $\bb_i \in \Gamma$ such that $\bb_i \ll \ba$ and  $b_{i i} = a_i$, so  we have $\ba = \lub(\bb_i;\;  i \in \Lambda)$.
\end{proof}

Given $j \in \setm$ let $H_j := \{ a \in \no \; | \; \exists  f \in \cap_{i = 1; \;   i \neq j}^m \U_{\r_i} \; \mbox{  such that  } \r_j (f) = a \}$ (i.e.\ $a \in H_j$ if and only if there exists $\ba \in \h$ having all entries equal to zero, except the $j$-th entry, which is equal to $a$). Then $H_j$  is a semigroup which has finite genus (since $\{\r_1, \ldots, \r_m\}$ is a complete set of \nw s). 

\begin{lemma} Let $\ba \in \Gamma$ and let $\Lambda = \{ j \; | \; a_j > 0\} \subset \setm$. If $\# \Lambda \geq 2$ then $a_j \notin H_j$ for all $j \in \Lambda$.
\end{lemma}
\begin{proof} Let $j \in \Lambda$ and assume by means of absurd that $a_j \in H_j$; let $\bb \in \no^m$ be the $m$-tuple having all entries equal to zero except the $j$-th, which is equal to $a_j$. Then $\bb \in \h$, $\bb \ll \ba$ and $\bb \neq \ba$, hence $\ba \notin \Gamma$.
\end{proof}

Let $\tilde{\Gamma} := \{ \ba \in \Gamma \; | \; \ba \mbox{ has at least two nonzero entries}\}$, an easy but important consequence of the above lemma is the following.

\begin{corollary} The set $\tilde{\Gamma}$ is finite.
\end{corollary}
\begin{proof} Let $G_j$ be the set of gaps of $H_j$, then  $\# (G_j) < \infty$ for all $j \in \setm$ and from the lemma above $\tilde{\Gamma} \subset G_1 \times \cdots \times G_m$.
\end{proof}

For each $\ba \in \Gamma$ let $f_\ba \in \R$ be such that $\br(f_\ba) = \ba$, and let $\B := \{ f_\ba \in \R; | \; \ba \in \G\}$.

\begin{proposition} \label{prop_span} The set $\B$ spans $\R$ as an $\F$-vector space.
\end{proposition}
\begin{proof}
We want to show that any $f \in \R\setminus\{0\}$ is a finite linear combination over $\F$ of elements of $\B$, and we do this by induction on the number of nonzero entries of $\ba:= \br(f)$. If this number is zero then $f \in \F^*$ and is a multiple of $f_{\mathbf{0}} \in \F^*$. Assume that $\ba$ has  $r$ nonzero entries, with $r \geq 1$, and for simplicity, let's assume that these entries are $a_1, \ldots,  a_r$. From Lemma \ref{lemma_lub} there are $\ba_1, \ldots, \ba_r \in \G$ such that $\ba = \lub(\ba_1, \ldots, \ba_r)$ and from Proposition \ref{prop_lub} there are $\l_1, \ldots, \l_r \in \F$ such that $g := \sum_{i = 1}^r \l_i f_{\ba_i}$ satisfies $\br(g) = \br(f)$. Since $\r_1(f) = \r_1(g) = a_1 > 0$ there is $\l \in \F^*$ such that $\r_1(f - \l g) < a_1$, moreover $\r_j(f - \l g) \leq a_j$ for all $j \in \{2, \ldots, m\}$. If $f = \l g$ we are done, otherwise we repeat the process, starting with $f  - \l g$ this time, until we get either that $f$ is a linear combination of finite elements of $\B$ or that the $m$-tuple obtained by applying the function $\br$ to  $f$ minus a finite linear combination of elements of $\B$ has less than $r$ nonzero elements (because the first entry is certainly zero); either way we're done.
\end{proof}

\begin{proposition} \label{fin_gen_alg} $\R$ is a finitely generated algebra over $\F$.
\end{proposition}
\begin{proof} Let $i \in \{1, \ldots, m\}$, we know that the semigroup $H_i \subset \no$ has finite genus, hence it is finitely generated, so let $H_i = \langle a_{i 1}, \ldots, a_{i r_i} \rangle$. For each $a_{i j}$ with $i \in \setm$ and $j \in \{1, \ldots, r_i\}$ there is $\ba_{i j} \in \G$ having all entries equal to zero, except the $i$-th entry which is equal to $a_{i j}$. Thus if $\ba \in \G \setminus \tilde{\G}$, i.e.\ if $\ba$ has only one positive entry, which is in the $i$-th position for some $i \in \setm$, then for certain $\alpha_1, \ldots, \alpha_{r_i} \in \no$ we have $\br(f^{\alpha_1}_{\ba_{i 1}} \cdot \ldots \cdot f^{\alpha_{r_i}}_{\ba_{i r_i}}) = \ba$ (recall that $f_{\ba_{i j}} \in \B$ and are such that $\br(f_{\ba_{i j}}) = \ba_{i j}$ for all $j \in \{1, \ldots, r_i\}$) and we can take $f_{\ba} := f^{\alpha_1}_{\ba_{i 1}} \cdot \ldots \cdot f^{\alpha_{r_i}}_{\ba_{i r_i}}$. Since $\tilde{\G}$ is a finite set, the result follows from the above proposition.
\end{proof}

\begin{theorem} \label{dim_finita} Let $f \in \R$, $f \neq 0$, then $\dim_{\F} \R/(f) < \infty$.
\end{theorem}
\begin{proof} We may assume that $f \in \R \setminus \F$. We also assume $m \geq 2$ (for $m = 1$ the proof is in \cite{matsumoto}). From Proposition \ref{prop_span} we have that the set $\overline{\B} := \{ \overline{f_\ba} \in \R/(f) \; | \; \ba \in \G\}$ spans $\R/(f)$ as a vector space over $\F$, and since $\tilde{\G}$ is finite, it suffices to show that for all $\ba \in \G \setminus \tilde{\G}$, except maybe for a finite number, we may take $f_\ba \in (f)$. Thus, we will show that for any $i \in \setm$ there exists $n_i \in \no$ such that for all $n \geq n_i$ with $n \in H_i$ we may find  $s \in (f)$ such that $\r_i(s) = n$ and $\r_j(s) = 0$ for all $j \in \setm$, $j \neq i$. 
For simplicity, let's take $i = 1$; we will consider two cases. In the first case, we assume that $\r_j(f) = 0$ for all $j = 2,\ldots, m$, hence $\r_1(f) > 0$ (since $f \notin \F$). Let $\ell_1$ be the largest gap of $H_1$ and set $d_1 := \r_1(f)$, if $a_{1 1}, \ldots, a_{1 r_1}$ are generators for $H_1$, then using the notation of the preceding proof, for any $n > \ell_1 + d_1$ we may find $\alpha_1, \ldots, \alpha_{r_1} \in \no$ such that $\r_1(f^{\alpha_1}_{\ba_{1 1}} \cdot \ldots \cdot f^{\alpha_{r_1}}_{\ba_{1 r_1}} f) = n$ and $\r_j(f^{\alpha_1}_{\ba_{1 1}} \cdot \ldots \cdot f^{\alpha_{r_1}}_{\ba_{1 r_1}} f) = 0$ for all $j = 2, \ldots, m$. In the second case we
 assume that there exists $j \in \{2, \ldots, m\}$ such that $\r_j(f) > 0$.
Let $g \in \R$ be such that $\r_1(g) > 0$ and $\r_i(g) = 0$ for all $i \in \{2, \ldots, m\}$ (such $g$ exists because the genus of $H_1$ is finite, moreover $g \notin \F$), then $\r_j (f g) \leq \r_j (f)$ and there exists $\l \in \F$ such that $\r_j(f g  - \l f) = \r_j (f (g - \l) < \r_j) (f)$. We have $g - \lambda \in \M_{\r_1}$ but for all $i \in \{2, \ldots, m\}$, since  $g \in \U_{\r_i}$ we have $g - \l \in \U_{\r_i}$ and $\r_i (f (g - \l)) \leq \r_i(f)$. By repeating this process we may find $h \in \M_{\r_1}$ such that $\r_i( h f ) \leq \r_i(f)$ for all $i \in \{2, \ldots, m\}$ and $\r_j (h f) = 0$; repeating even further we find $t \in \M_{\r_1}$ such that $\r_i(t f) = 0$ for all $i \in \{2, \ldots, m\}$ (observe that $\r_1(t f ) > 0$ since if $\r_1(t f) = 0$ then $t f \in \cap_{i = 1}^m \U_{\r_i} = \F$, and a fortiori $f \in \F$, a contradiction). Let $\ell_1$ be the largest gap in $H_1$ and set $d_1 := \r_1(t f)$;  given $n > \ell_1 + d_1$  let $u \in \M_{\r_1}$ be such that $\r_1(u) = n - d_1$ and $\r_i (u) = 0$ for all $i \in \{2, \ldots, m\}$, then $\r_1( u t f  ) = n$ and $\r_i (u t f) = 0$ for all $i \in \{2, \ldots, m\}$. This completes the proof.
\end{proof}

We have already observed that $\R$ is a domain, and we will denote by $\FF$ its field of fractions.

\begin{lemma} $\FF$ is an algebraic function field of one variable over $\F$.
\end{lemma}
\begin{proof} 
Let $f \in \R$, $f \neq 0$, from Theorem \ref{dim_finita} we know that $\R/(f)$ is an $\F$-vector space of finite dimension, furthermore, all ideals of $R/(f)$ are $\F$-subspaces hence $\R/(f)$ is an artinian ring, so $\dim_{\mathrm{Krull}} \R/(f) = 0$. Taking $f \in \R \setminus \F$,  from \cite[Corollary 13.11]{eisenbud} we have $\dim_{\mathrm{Krull}} \R = \dim_{\mathrm{Krull}} \R/(f) + 1$; on the other hand from \cite[Theorem A, page 223]{eisenbud} we get $\mbox{tr deg}_{\F}\FF  = \dim_{\mathrm{Krull}} \R = 1$. 
\end{proof}

\begin{corollary} The algebra $\R$ is the affine coordinate ring of an (irreducible) algebraic curve.
\end{corollary}
\begin{proof} It is an immediate consequence of  Proposition \ref{fin_gen_alg} and the above lemma.
\end{proof}

Let $i \in \setm$, from the proof of Theorem \ref{dim_finita} we get that if $f \in \R \setminus \{0\}$ then there exists $g \in \M_{\r_i}$ such that $ g f \in \M_{\r_i}$, hence if $a,  b \in \R \setminus \{0\}$  and $g_1 a,  g_2 b \in \M_{\r_i}$ with $g_1, g_2 \in \M_{\r_i}$ then $(g_1 g_2) a, (g_1 g_2) b \in \M_{\r_i}$. 

\begin{definition} Let  $i \in \setm$ and let $v_i: \FF \ra \Z \cup \{ \infty\}$ be the function defined by setting $v_i(0) : = \infty$ and $v_i (a/b) := \r_i (g b) - \r_i (g a)$, where $a, b, \in \R \setminus \{0\}$  and $g  \in \M_{\r_i}$ is such that $g a, g b \in \M_{\r_i}$. 
\end{definition}

Observe that $v_i(a/b)$ does not depend on the choice of $g$ because if $h \in \M_{\r_i}$ is such that $h a, h b \in \M_{\r_i}$ then $\r_i(g b ) - \r_i(g a ) - (\r_i(h b) - \r_i(h a)) = \r_i (g b h a ) - \r_i (g a h b) = 0$, for all $i \in \setm$; a similar reasoning shows that if $a'/b'= a/b$, with $a, a', b, b' \in \R\setminus\{0\}$ then $v_i(a/b) = v_i(a'/b')$.

\begin{lemma} Let $i \in \setm$. \\
a) The function $v_i : \FF \ra \Z \cup \{\infty\}$ is a discrete valuation of the function field $\FF \;| \; \F$; \\
b) If $f \in \R$ then $v_i(f) \geq 0$ when $f \in \U_{\r_i}$ and $v_i(f)  = - \r_i(f)$ when $f \in \M_{\r_i}$. 
\end{lemma}
\begin{proof} Given $f, g \in \FF\setminus \{0\}$ it is easy to check that $v_i (f g) = v_i(f) + v_i(g)$ and that $v_i(f) = 0$ if $f \in \F^*$. Since $H_i$ has finite genus, we know that for a sufficiently large $n \in \N$  there are $f, g \in \M_{\r_i}$ such that $\r_i(f) = n$, $\r_i(g) = n + 1$, hence $v_i (f / g) = 1$. Let $f = a/b, g = c/d \in \FF$, with $a, c \in \R$ and $b, d \in \R\setminus\{0\}$, and let $h_1, h_2 \in \M_{\r_i}$ such that $h_1 a,  h_1 b, h_2 c, h_2 d \in \M_{\r_i}$, then $v_i (f + g) = v_i ((a d + b c) /bd) = \r_i ( h_1 h_2 b d ) - \r_i (h_1 h_2 a d + h_1 h_2 b c) \geq  \min \{\r_i ( h_1  h_2 b d ) - \r_i ( h_1 h_2 a d ),\r_i ( h_1 h_2 b d ) - \r_i ( h_1 h_2  b c )\} = \min \{\r_i ( h_1  b  ) - \r_i ( h_1  a  ),\r_i ( h_2  d ) - \r_i ( h_2   c )\} = \{ v_i (f), v_i(g)\}$.

Now let $f \in \R\setminus \{0\}$ and $g \in \M_{\r_i}$ be such that $g f \in \M_{\r_i}$,  if $f \in  \U_{\r_i}$ then from (N5) and the fact that $\r_i$ is normalized we get $v_i(f/1) = \r_i(g) - \r_i(g f) \geq - \r_i(f) = 0$; on the other hand, if $f \in \M_{\r_i}$ then $v_i(f/1) = \r_i(g) - \r_i(g f) = - \r_i(f)$.
\end{proof}

This shows that every \nw\ $\r_i$ on $\R$ defines a valuation $v_i$ of the function field $\FF \; | \; \F$. These are distinct valuations (e.g. for a sufficiently large $n \in \N$ we may find $f_i \in \M_{\r_i}$ for all $i \in \setm$ such that $v_i (f_i) = - n$ and $v_j (f_i) \geq 0$ for all $j \in \setm \setminus \{i\}$). We denote by $P_i$ the place associated to the valuation $v_i$ and be $\O_{P_i}$ the corresponding valuation ring ($i \in \setm$).

\begin{proposition} For all $i \in \setm$ the place $P_i$ has degree one (a fortiori, $\F$ is the full field of constants of $\FF$).
\end{proposition}
\begin{proof} Let $i \in \setm$, we must prove that the inclusion map $\F \ra \O_{P_i}/P_i$ is surjective. Let $f = a/b \in \O_{P_i}$, where $a, b \in \R$, let $g \in \M_{\r_i}$ such that $g a, g b \in \M_{\r_i}$ and assume that $v_i (f) = 0$. Then $\r_i (g b ) = \r_i (g a)$ and there exists a unique $\l \in \F^*$ such that $\r_i (g a - \l g b) < \r_i (g b)$. Let $h \in \M_{\r_i}$ be such that $ h (a - \l b), h b  \in \M_{\r_i}$, then $v_i (a/b - \l) = \r_i(h b) - \r_i (h (a - \l b)) = \r_i(h g b) - \r_i (h g (a - \l b))$, so from $\r_i (g b) - \r_i (g a - \l g b) > 0$ and property (N3) we get $\r_i(h g b) - \r_i (h g (a - \l b)) > 0$, which completes the proof.
\end{proof}

We denote by $\P(\FF)$ the set of places of the function field $\FF \; | \; \F$. For  $P \in \P(\FF)$ we write $\O_P$ for the corresponding valuation ring; let $\S(\R) := \{ P \in \P(\FF) \; | \; \R \subset \O_P\}$.

\begin{proposition} $\S(\R) = \P(\FF) \setminus \{P_1, \ldots, P_m\}$. 
\end{proposition}
\begin{proof} First we observe that, for all $i \in \setm$ we have $P_i \notin \S(\R)$, since $\R \subset \O_{P_i}$ would imply $\M_{\r_i} = \emptyset$, a contradiction with the fact that $\r_i$ is non-trivial.
Suppose by means of absurd that $\P(\FF) \setminus (\S(\R) \cup \{P_1, \ldots, P_m\}) \neq \emptyset$. Then, from the Strong Approximation Theorem (see \cite[Thm. I.6.4]{sti}) we know that for all $j \in \N$  there exists $f_j \in \FF$ such that $v_i(f_j) = j$,  for all $i \in \{1, \ldots, m\}$ and $f_j \in \O_Q$ for all $Q \in \S(\R)$, thus $f_j \in \cap_{Q \in \S(\R)} \O_Q =: \bar{\R}$, the integral closure of $\R$ in $\FF$. Let $W := \{ x \in \bar{\R} \; |\; v_i(x) > 0 \; \forall \; i = 1, \ldots, m\}$, observe that $W$ is an $\F$-vector space and also $W \cap \R = \{ 0 \}$: in fact, if $ x \in W \cap \R$ then $\r_i(g) - \r_i(g x) > 0$ for some $g \in \M_{\r_i}$, thus $\r_i(g x) < \r_i(g)$ and from (N5) either $\r_i(x) = 0$ for all $i \in \setm$ or $x = 0$, since $\cap_{i = 1}^m \U_{\r_i} = \F$ and $x \in W$  we must have $x = 0$. Thus $\dim_\F W \leq \dim_\F \bar{\R}/\R$ and this last dimension is finite (see e.g. \cite[Lemma 8]{matsumoto}), but $\{f_1, \ldots, f_n\} \subset W$ is a linearly independent set over $\F$ for all $n \in \N$.
\end{proof}

\begin{corollary} \label{main}  $\R$ is an $\F$-algebra admitting a complete set of $m$ \nw s\  if and only if $\R$ is the ring of regular functions of an affine geometrically irreducible algebraic curve, whose points in the closure have a total of $r$ branches, all of them corresponding to rational places in the field of rational functions of the curve.
\end{corollary}
\begin{proof} The ``only if'' part is a consequence of the above results. As for the ``if'' part let $\X$ be the affine curve and $\overline{\X}$ be its closure, if $\Y$ is the normalization of $\overline{\X}$ and $\eta: \Y \rightarrow \overline{\X}$ is the normalization morphism then there are $m$ rational  points $Q_1, \ldots, Q_m$ in the inverse image by $\eta$ of the set $\overline{\X} \setminus \X$. Now we proceed as in theorem \ref{ida}; thus we  observe that   
$\R = \cap_{Q \in \X}  \O_Q$, where $\O_Q$ is the local ring at $Q \in \X$ and denoting by $v_k$ the discrete valuation of $\F(\overline{\X})$ associated to $Q_k$ ($k \in \{1, \ldots,m\}$) we define  the function $\r_k : \R \rightarrow \no \cup \{\infty\}$   by setting  $\r_k(0) : = - \infty$, $\r_k(f) := 0$ if $v_k(f) \geq 0$ and $\r_k(f) := - v_k(f)$ if $v_k(f) < 0$, for all $f \in \R \setminus \{0\}$, one may check that $\r_k$  is an \nw\ for all $k \in \setm$. From $\cap_{k = 1}^m \U_{\r_k} = \R \cap (\cap_{k = 1}^m \O_{Q_k}) = \F$  and the fact that $\S_k := \rho_k(\cap_{1\leq i \leq m\, ; \,  i \neq k} \cU_{\r_i}) = \r_k(\cap_{Q  \in \X} \O_Q) $ 
is the Weierstrass semigroup at $Q_k$ for all $k \in \setm$  we get that $\{\r_1 , \ldots, \r_m\}$ is a complete set of \nw s for $\R$.
\end{proof}

\begin{theorem} 
Let $\R$ be an $\F$-algebra that admits a complete set of $m$  \nw s, let $\varphi: \R \rightarrow \F^n$ be a surjective morphism of $\F$-algebras and $\ba \in \no^m$, then $C(\ba)$ is an algebraic-geometric Goppa code $C_{\cal{L}}(D,G)$ with $G$ supported on $m$ points.
\end{theorem}
\begin{proof}
From the hypothesis on $\R$ we know that there is a geometrically irreducible, projective, nonsingular  curve ${\mathcal{Y}}$ and points $Q_1, \ldots, Q_m$ such that $\R = \cap_{P \in \mathcal{Y} \setminus \{Q_1, \ldots, Q_m\}} {\cal O}_P$. For $i \in \{1, \ldots, n\}$ consider the $\F$-algebra surjective homomorphism $\pi_i: \F^n \rightarrow \F$ defined by $\pi_i(\lambda_1, \ldots, \lambda_n) = \lambda_i$, then $M_i := (\pi_i \circ \varphi)^{-1}(0)$ is a maximal ideal of $\R$. Furtermore, for distinct  $i, j \in \{1, \ldots, n\}$ we get $M_i \neq M_j$ since $\varphi$ is surjective and then exists $g_{i j} \in \R$ such that $(\pi_i\circ \varphi)(g_{i j}) = 0$ and $(\pi_j\circ \varphi)(g_{i j}) \neq 0$. From \cite[Prop. III.2.9]{sti} we get that there are $P_1, \ldots,P_n \in {\cal Y}$ such that $P_i \notin \{Q_1, \ldots, Q_m\}$,   $M_i = {\cal M}_{P_i} \cap \R$ (where ${\cal M}_{P_i}$ is the maximal ideal of ${\cal O}_{P_i}$) for all $i = 1, \ldots, n$. We also get $\F \simeq \R/M_i \simeq  {\cal O}_{P_i}/{\cal M}_{P_i}$ for all $i = 1, \ldots, n$ hence $P_1, \ldots, P_n$ are rational points of ${\cal Y}$ and we may rewrite $\varphi$ as the morphism over ${\cal O}_{P_1}/{\cal M}_{P_1} \times \cdots \times {\cal O}_{P_n}/{\cal M}_{P_n}$ defined by $\varphi(f) = (f + {\cal M}_{P_1}, \ldots, f + {\cal M}_{P_n})$.
Let $G := a_1 Q_1 +\cdots + a_m Q_m$, then $L(G) \subset \R$ and
 \begin{align*} 
L(G)\; = \;\; & \{ f \in \R : v_i(f) + a_i \geq 0 \textrm{\  for all \  } i = 1, \ldots, m\} = \\
       &  \{ f \in \R : - v_i(f) \leq  a_i  \textrm{\  whenever \ } v_i(f) < 0, i = 1, \ldots, m\} =  \\
       &  \{ f \in \R : - v_i (f) \leq a_i \textrm{\ whenever\  } f \in \M_{\r_i},  i = 1 ,\ldots, m\}  = \\
       &  \{ f \in \R : \r_i (f) \leq a_i \textrm{\ whenever\  } f \in \M_{\r_i},  i = 1 ,\ldots, m\}  = \\
       & \{ f \in \R : \r_i (f) \leq a_i,   i = 1 ,\ldots, m\}  = \cL(\ba),
       \end{align*} 
 hence $C(\ba) = C_{\cal{L}}(D, G)$, where $D = P_1 + \cdots + P_n$.
\end{proof}

\end{document}